\documentclass[conference]{IEEEtran}
\makeatletter
\def\ps@headings{
\def\@oddhead{\mbox{}\scriptsize\rightmark \hfil \thepage}
\def\@evenhead{\scriptsize\thepage \hfil \leftmark\mbox{}}
\def\@oddfoot{}
\def\@evenfoot{}}
\makeatother
\usepackage{textcomp}
\usepackage{amsmath}
\usepackage{amssymb}
\usepackage{float}
\usepackage{psfrag}
\usepackage[dvips]{graphicx,color}
\usepackage{cite}
\usepackage{epstopdf}

\newtheorem{theorem}{Theorem}[section]

\newtheorem{remark}[theorem]{Remark}

\newenvironment{proof}[1][Proof]{\begin{trivlist}
\item[\hskip \labelsep {\bfseries #1}]}{\end{trivlist}}

\begin{document}

%Theorem definitions made here
\newtheorem{lem}{Lemma}
\newtheorem{thm}{Theorem}
\newtheorem{prb}{Problem}

%Title Defined Here
\title{Offline and Online Incentive Mechanism Design for Smart-phone Crowd-sourcing}
\author{ Ashwin Subramanian \quad G Sai Kanth \quad Rahul Vaze }
\maketitle

%The abstract starts here
\begin{abstract}

In this paper, we consider the problem of incentive mechanism design for smart-phone crowd-sourcing.
Each user participating in crowd-sourcing submits a set of tasks it can accomplish and its corresponding bid. 
The platform then selects the users and their payments to maximize its utility while ensuring truthfulness, individual rationality, profitability, and polynomial algorithm complexity. 
Both the offline and the online scenarios are considered, where in the offline case, all users submit their profiles simultaneously, while in the online case they do it sequentially, and the decision whether to accept or reject each user is done instantaneously with no revocation.
The proposed algorithms for both the offline and the online case are shown to satisfy all the four desired properties of an {\it efficient} auction. 
Through extensive simulation, the performance of the offline and the online algorithm is also compared.

\end{abstract}
%The abstract ends here

%This section provides an introduction to the concept of smart-phone crowd-sourcing
\section{Introduction}
Crowd-sourcing using smart phones is a new idea that has gained widespread interest 
\cite{Nericell, Sensorly, PotholeCrSourcing2008, EarphoneCrSourcing2010, CrowdSourcingPractical2009, Sheng2012CrSourcingEnergy, WalrandCrSourcing2012, M-Sensing2012}. 
Smart phones these days are equipped with multiple sensors that can be used to monitor key features of the surrounding environment that help in improving the user experience or simplifying human effort. 
Collectively using data derived from multiple smart phones (called {\it crowd-sourcing}) helps in improving the social welfare, e.g. 
helps public utility companies to track potholes locations, electricity failure, emergency relief operations, traffic congestion etc.

Several commercial applications using smart phone crowd-sourcing are already in place, such as Sensorly \cite{Sensorly}, Nericell \cite{Nericell}, 
Google voice recognition and Apple's Siri also use data from users spread across many different locations to improve their services. 
In some applications, users volunteer to share their data since that also helps in improving their own {\it utility}. 
This is, however, not true in general, and necessitates an incentive mechanism design, where users are externally incentivized in the form of payments for the data/tasks they are willing to share/perform.

In this paper, we consider the incentive mechanism design problem for smart phone crowd-sourcing, where we model it as a {\it reverse auction}.
The platform announces a set of tasks that it wants to accomplish, and each user submits the list of tasks it is ready to provide and its corresponding bid. 
The platform has a utility function associated with its set of tasks, and the problem is to find the set of users and their corresponding payments that maximizes its utility. 
One of the main challenges in designing such incentive mechanism design is to ensure {\it truthfulness}, i.e. no user should not have any incentive to bid more than its true valuation. 
Since users expend some resources to accomplish tasks, it is natural to assume that they will seek to maximize the profit they intend to make from the platform.

One of the ways to ensure truthful auction is the VCG mechanism \cite{VCG1961}. 
In a forward auction, the VCG mechanism charges each individual the harm it causes to other bidders in terms of the social welfare utility, 
while in a reverse auction, it pays each user an amount equal to the value contributed by the user to the auction \cite{VCG1961}.
However, finding the winning set of users in a VCG mechanism is combinatorial and has exponential complexity. 
There also some technical difficulties with the VCG mechanism \cite{VCGIssues}. 
So for a computationally feasible operation such as smart-phone crowd-sourcing, one cannot directly use the VCG mechanism. 
Several variations of VCG mechanism can be found in \cite{OnlineVCG, AuctionSurvey}.

There are two basic paradigms for smart-phone crowd-sourcing, {\it offline} and {\it online}. 
In offline case, all users are present/active simultaneously, and send their profiles to the platform at the same time. 
In the online case, users arrive sequentially one at a time and submit their profiles, and the platform must decide immediately whether to accept or reject the user and how much to pay the user. 
A decision once made, is irrevocable. 
For example, the offline scenario is applicable for current traffic congestion monitoring, while the online case is more suited for potholes tracking type of applications that are localized, 
where users pass over potholes in a given area sequentially.
The online scenario is more general than the offline case, since all potentially participating users may not be active at the same time.

For both the offline and online scenarios, the platform's objective is to select the set of users and their payments to maximize its utility, 
subject to the following four requirements \cite{NarahariMechanismDesignTutorialI2008, NarahariMechanismDesignTutorialII2008}: {\it computational efficiency} - the algorithm implemented by the platform has polynomial run time complexity, 
{\it individual rationality} - the selfish utilities of all users involved are non-negative, {\it profitability}- the platform utility is non-negative after the auction concludes,
and {\it truthfulness} - no user has any incentive to bid different from its true valuation.

In prior work, an incentive design mechanism called {\it M-Sensing} satisfying all the four properties for the offline case has been derived in \cite{M-Sensing2012}. 
M-sensing uses a greedy algorithm that at each step adds the user with maximum incremental utility.
It pays each user in the selected set, the maximum value which that user can bid and still be selected at some possible position in the greedy selection phase. More recently in \cite{Luo2013CrowdSourcing}, both offline and online algorithms for sensing time schedules have been proposed for the crowd-sourcing problem. The proposed algorithms are shown to be truthful, and more importantly, analytical performance guarantees have been found on the performance of both offline and online algorithms in \cite{Luo2013CrowdSourcing}. However,  \cite{Luo2013CrowdSourcing} only considers a linear utility function.

In this paper, we first propose an algorithm for the offline case motivated from the VCG mechanism design and the greedy approach of the M-sensing algorithm.
Our algorithm called SMART, first greedily finds a screening set that is identical to the M-sensing winner set. Thereafter, inspired by the VCG mechanism, the screened set is refined further to keep only those users that have positive marginal utility, and the payment to each user is equal to the marginal utility increase compared to the next best user plus the bid of the next best user.
We show that SMART satisfies all the four properties required for efficient mechanism design. 
In addition, we show that the platform utility of SMART is always greater than or equal to that of M-sensing, while with incurring identical complexity.

Another important advantage of SMART algorithm is that it can be adapted easily for the online case, while still satisfying the four properties. 
For the online case, we take motivation from the online $k$-secretary problem \cite{LleinbergK-Secretary2005, babaioff2007k-Secretary}. 
In the online $k$-secretary problem, $N$ secretaries with arbitrary ranks arrive in a uniformly random order and the problem is to select the $k$ best ranked secretaries in an online manner. 
The best known algorithms for solving the online $k$-secretary problem reject the first $m = N/e$ secretaries, and generating a threshold set from the first $m$ secretaries, 
which is then used to select the $k$ best secretaries among the remaining $N-m$ secretaries. 
For the online smart-phone crowd-sourcing problem, we reject the first few users and run the offline SMART algorithm on their profiles. The output of the offline SMART algorithm is then used to select the users from the remaining users and to decide their payments, such that the four properties are satisfied. 

Typically, the online and offline solutions are compared through the competitive ratio \cite{BorodinOnlineBook}, 
which is defined to be the ratio of the utility obtained by the online  algorithm to that obtained by the offline algorithm. 
Ideally, we would like the competitive ratio to be as close to unity as possible. 
For the $k$-secretary problem, the competitive ratio has been shown to be $1-1/e$ if $m=\frac{N}{e}$. 
In the crowd-sourcing problem, as shown in \cite{LleinbergK-Secretary2005, babaioff2007k-Secretary}, if you pay each selected user its own bid, the mechanism is not truthful. 
This key point complicates the computation of the competitive ratio for the crowd-sourcing problem.
Users selected in the offline and online case are actually paid differently to ensure truthfulness. 
So there is no easy analytical way to compare the utility of the offline and the online algorithms, unlike the $k$-secretary problem. 
We thus resort to extensive simulation to find the competitive ratio that depends of $m$, the number of users rejected and used for running the offline SMART algorithm, and find that most often $m = n/3$, where $n$ is the total number of users. 

%Describes and Formulates the problem
\section{System Model}

The problem of smart-phone crowd-sourcing is modeled as a reverse auction.
The platform declares a set of tasks to be performed, of which each task has some value to the platform.
The users submit their profiles to the platform, where each profile contains a list of tasks they will complete and their corresponding bid.
The users expect to be paid an amount at least equal to their bid if they perform their tasks.
In the offline case, the platform receives the profiles of all users simultaneously.
The platform must decide which users to select and how much to pay each selected user.
The offline problem can therefore be described as follows.
Given a set $U$ of users, select a set $T\subseteq U$, so that the platform utility\footnote{The term platform utility is defined formally later. It relates to the profit the platform obtains from a set of users}
 of $T$ is maximized over all possible subsets of $U$.
In the online case, users arrive one at a time and submit their profiles to the platform.
The platform must decide immediately whether to accept or reject the user and how much to pay the user. 
A decision once made, is irrevocable.
Similar to the offline case, the objective in this case too is to maximize the platform utility of a set $T\subseteq U$, however, with causal user profiles. 

Let the platform declare the set of tasks $\Gamma = \{t_1,t_2,t_3,\ldots ,t_m\}$, where the value of a task $t_k$, is given by $\chi(t_k)$, which for the sake of brevity will be written as $\chi_k$.
Further, the function $\chi$ is extended for any set of tasks $\tau \subseteq \Gamma$, $\chi(\tau)$, where it could be linear $\chi(\tau) = \sum_{i: t_i \in \tau} \chi_i$ or any other arbitrary combinatorial function.

%This table lists Commonly Used Notations
\begin{table}
\caption{Notation Used}
\begin{tabular}{c l}
\hline
Notation & Description\\
\hline
$U$ & Set of all users\\
$S$ & Set of screened users\\
$T$ & Set of winning users\\
$i,j$ & User $i$ and User $j$ in set $U$\\
$n$ & Total number of users\\
$\Gamma$ & The platform's set of tasks\\
$t_k$ & Task $k$ in set $\Gamma$\\
$m$ & Total number of sensing tasks\\
$\chi_k$ & Value of task $k$\\
$\chi(\tau)$ & The value of tasks in $\tau$\\
$v_i$ & Value of all tasks performed by user $i$\\
$v(S)$ & Value of tasks performed by users in $S$\\
$v_i(S)$ & Marginal value of user $i$, given set $S$\\
$b_i$ & Bid of user $i$\\
$p_i$ & Payment made to user $i$\\
$c_i$ & Cost incurred by user $i$\\
$\tau_i$ & Set of tasks performed by user $i$\\
$\tau(S)$ & Set of tasks performed by users in $S$\\
$\tau_i(S)$ & Marginal tasks done by user $i$, given set $S$\\
$\omega_i$ & Personal or intrinsic utility of user $i$\\
$u_i$ & Platform utility of user $i$\\
$u(S)$ & Platform utility of users in $S$\\
$u_i(S)$ & Marginal utility of user $i$, given set $S$\\
$\gamma_i$ & Minimum replacement user's bid\\
$\sigma_i(S)$ & Difference of marginal value and bid\\
\hline
\end{tabular}
\end{table}

Let the set of users be $U = \{1,2,3,\ldots ,n\}$.
Each user $i \in U$ can perform a set of tasks $\tau_i \subseteq \Gamma$ for a bid $b_i$.
The value of a user $i$ to the platform is, $v_i = \chi(\tau_i)$.
 
For a set of users $S\subseteq U$, the set of tasks performed by them are denoted by $\tau(S)$, with
\begin{equation}\label{tauS_in_terms_of_users_in_S}
\tau(S) = \bigcup_{i\in S} \tau_i,
\end{equation}
and the value of $S\subseteq U$ is,
\begin{equation}\label{vS_in_terms_of_tauS}
v(S) = \chi(\tau(S)).
\end{equation}
Note that the  function $v(S)$ need not be linear.

Let the marginal value of a user $i$ with respect to a set $S\subseteq U$ be denoted by $v_i(S)$, defined as,
\begin{equation}\label{marginalvalue_expression}
v_i(S) = v(S\cup\{i\}) - v(S).
\end{equation}  
The marginal value of a user $i$ with respect to a set $S$ is also the value of task set $\tau_i\big\backslash\big(\tau_i\cap\tau(S)\big)$.

The platform utility of a user $i$ is denoted by $u_i=v_i-p_i$, and is defined as the difference between the value of user $i$ and the payment made to user $i$, which measures the profit the platform obtains by choosing user $i$.

Generalizing, the platform utility of a set $S$ of users is denoted by $u(S)$ and is, 
\begin{equation}\label{utility_of_set_S_to_platform}
u(S) = v(S) - \sum_{i\in S} p_i.
\end{equation}

The marginal utility, $u_i(S)$, of a user $i$, given a set $S$ of users is the difference between the utility of the set $S \cup \{i\}$ and the utility of set $S$,  
\begin{equation}\label{marginalutility_expression}
u_i(S) = u(S\cup \{i\}) - u(S).
\end{equation}
The marginal utility is also equal to the difference between the marginal value of user $i$ and the payment made to user $i$,
\begin{equation}\label{marginalutility_user}
u_i(S) = v_i(S) - p_i.
\end{equation}

The difference between the marginal value (with respect to $S\subseteq U$) and bid of a user $i$ is denoted by $\sigma_i(S)$,
\begin{equation}
\sigma_i(S) = v_i(S\backslash\{i\}) - b_i.
\end{equation}

The personal or selfish utility of a user $i$ is denoted by $\omega_i$, and measures the personal profit that user $i$ has gained through the reverse auction.
If $c_i$ is the cost incurred by user $i$ in completing tasks for the platform and $p_i$ is the payment made to user $i$ by the platform, the personal utility of user $i$ is
\begin{equation}\label{personal_utility}
\omega_i = p_i - c_i,
\end{equation}
if user $i$ is selected by the platform and zero otherwise. 

The problem of smart-phone crowd-sourcing is now formally defined.
\begin{prb}
Given a set $U$ of users,
\begin{equation}
\max_{T\subseteq U} u(T),
\end{equation} subject to the four properties of an efficient auction, namely, computational efficiency, individual rationality, profitability and truthfulness.
\end{prb}

The four properties of an efficient auction are described below.
\begin{itemize}
\item \emph{Computational efficiency - } An auction is computationally efficient if the algorithm used to find winning users has polynomial run time complexity.
\item \emph{Individual rationality - } An auction is individual rational if the selfish utilities of all users involved are non-negative.
\item \emph{Profitability - } An auction is profitable if the platform utility is non-negative after the auction concludes. 
\item \emph{Truthfulness - } An auction is called truthful, if and only if a user involved in the auction has no incentive to bid different from its true value.
\end{itemize}  

The first three properties ensure that the proposed algorithms are feasible.
Truthfulness makes the reverse auction free from market manipulation. 
It establishes that there is no incentive for users to manipulate their bids in the hope of higher individual profits. 

%This section discusses the improvements made to the current offline solution to the smart-phone crowd-sourcing problem
\section{Offline Case}

\subsection{Design of the Mechanism}
To solve the offline smart-phone crowd-sourcing problem, we propose an algorithm called SMART, that stands for Search for Marginal Appropriate Replacement Tasks users.
The algorithm consists of three phases, the screening phase followed by the winner selection phase and finally the bad user removal phase. 
In the winner selection phase, the payment made to each winning user is also determined.

%The SMART Algorithm and associated functions
\begin{table}
\begin{tabular}{r l}
\hline
& \textbf{SMART Algorithm}\\
\hline
1 	&	\textbf{// Phase 1 - User Screening}\\
2 	&	$S \leftarrow \emptyset$, $P \leftarrow \{p_1,p_2,p_3,\ldots, p_n\}$\\
3	&	$i \leftarrow \arg \max_{j\in U}(v_j(S)-b_j)$\\
4 	& 	\textbf{while} $v_i(S) > b_i$ and $S \neq U$ do\\
5 	& 	\quad $S \leftarrow S \cup \{i\}$\\
6 	& 	\quad $i \leftarrow \arg \max_{j\in U}(v_j(S)-b_j)$\\
7 	& 	\textbf{endwhile}\\
8 	& 	\textbf{for} each $i \in U$ do\\ 
9	&	\quad $p_i \leftarrow 0$\\
10	&	\textbf{endfor}\\
11	& 	\textbf{// Phase 2 - Winner Selection}\\
12	&	$T \leftarrow S$\\
13	& 	\textbf{for} $i = 1,2,\ldots,|T|$ do\\
14	&	\quad $(j,\gamma_i,\beta_j) \leftarrow $ Next Best User $(i,U,T)$\\
15	&	\quad $\sigma_i(T) \leftarrow v_i(T\backslash\{i\}) - b_i$\\
16	&	\quad $\beta_i \leftarrow \text{User Entry Payment}(U,i)$\\
17	&	\quad \textbf{// Cond 1 - Positive marginal utility}\\
18	&	\quad \textbf{if} $\sigma_i(T) > 0$ then\\
19	&	\quad \quad \textbf{// Cond 1.1 - Pay next best user's bid}\\
20	&	\quad \quad \textbf{if} $\gamma_i - b_i \geq 0$ and $\gamma_i \leq \beta_i$ and $\gamma_i \neq \infty$\\
21	&	\quad \quad \quad $p_i \leftarrow \gamma_i$\\
22	&	\quad \quad \textbf{// Cond 1.2 - Replace with next best user}\\
23	&	\quad \quad \textbf{else if} $\gamma_i < b_i$ and $\gamma_i \leq \beta_i$ and $\gamma_i \neq \infty$\\
24	&	\quad \quad \quad $(T,p_j) \leftarrow$ Replace User $(U,T,i,j)$\\
25	&	\quad \quad \textbf{// Cond 1.3 - Pay marginal value in S}\\
26	&	\quad \quad \textbf{else if} $\gamma_i > \beta_i$ or $\gamma_i = \infty$\\
27	&	\quad \quad \quad $p_i \leftarrow \min(\sigma_i(T) + b_i,\beta_i)$\\
28	&	\quad \quad \textbf{endif}\\
29	&	\quad \textbf{// Cond 2 - Non-Positive Marginal Value}\\
30 	&	\quad \textbf{else if} $\sigma_i(T) \leq 0$ and $\gamma_i \neq \infty$\\
31	&	\quad \quad $(T,p_j) \leftarrow$ Replace User $(U,T,i,j)$\\
32	&	\quad \textbf{else if} $\sigma_i(T) \leq 0$ and $\gamma_i = \infty$\\
33	&	\quad \quad $T \leftarrow T \backslash \{i\}$\\
34	&	\quad \textbf{endif}\\
35	&	\textbf{endfor}\\
36	&	\textbf{// Phase 3 - Bad User Removal}\\
37	&	\textbf{for} each $i$ in $T$ do\\
38	&	\quad \textbf{if} $u_i(T\backslash\{i\}) \leq 0$ then\\
39	&	\quad \quad $T \leftarrow T \backslash \{i\}$\\
40	&	\quad \quad $p_i \leftarrow 0$\\
41	&	\quad \textbf{endif}\\
42	&	\textbf{endfor}\\
43 	&	Return $(T,P)$\\
\hline
\end{tabular}
\end{table}

\begin{table}
\begin{tabular}{r l}
\hline
& \textbf{Next Best User $(i,U,T)$}\\
\hline
1	&	$j \leftarrow \arg\max_{k\in U\backslash T} v_k(T\backslash \{i\}) - b_k$, $\beta_j \leftarrow \infty$\\
2	&	\textbf{if} $v_j(T\backslash \{i\}) - b_j > 0$ then\\
3	&	\quad $\gamma_i \leftarrow v_i(T\backslash\{i\}) - v_j(T\backslash\{i\}) + b_j$\\
4	&	\textbf{else}\\
5	&	\quad $j \leftarrow -1$, $\gamma_i \leftarrow \infty$\\
6	&	\textbf{endif}\\
7	&	Return $(j,\gamma_i,\beta_i)$\\
\hline
\end{tabular} 
\end{table}

\begin{table}
\begin{tabular}{r l}
\hline
& \textbf{Replace User $(U,T,i,j)$}\\
\hline
1	&	$T \leftarrow (T \backslash \{i\}) \cup \{j\}$\\
2	&	$(k,\gamma_j,\beta_k) \leftarrow $ Next Best User $(j,U,T)$\\
3	&	\textbf{if} $\gamma_j < v_j(T\backslash\{j\})$ then\\
4	&	\quad $p_j = \gamma_j$\\
5	&	\textbf{else if} $\gamma_j \geq v_j(T\backslash\{j\})$\\
6	&	\quad $p_j = v_j(T\backslash\{j\})$\\
7	&	\textbf{endif}\\
8	&	Return $(T,p_j)$\\
\hline
\end{tabular} 
\end{table}

\begin{table}
\begin{tabular}{r l}
\hline
& \textbf{User Entry Payment $(U,i)$}\\
\hline
1 	&	$S_i \leftarrow \emptyset$, $\beta_i \leftarrow 0$, $j \leftarrow \arg\max_{k\in U\backslash\{i\}}(v_k(S)-b_k)$\\
2 	& 	\textbf{while} $v_i(S_i) > b_i$ and $S_i \neq U\backslash\{i\}$ do\\
3 	& 	\quad $\beta_i \leftarrow \max(\beta_i,v_i(S_i)-v_j(S_i)+b_j)$\\
4 	& 	\quad $S_i \leftarrow S_i \cup \{j\}$;\\
5 	& 	\quad $j \leftarrow \arg\max_{k\in U\backslash\{i\}}(v_k(S)-b_k)$\\
6 	& 	\textbf{endwhile}\\
7	&	Return $\beta_i$\\
\hline
\end{tabular} 
\end{table}

%Explanation of the algorithm
The algorithm follows the greedy approach while selecting users in the screening phase.
It maintains a screening set $S$ (initially set to $\emptyset$), to keep track of screened users in this phase. 
It iterates through $U\backslash S$ and picks the user $i$, having the maximum difference in marginal value and bid $v_i(S) - b_i$, with respect to current set $S$. 
This user is added to $S$ and the process repeats as long as there are users with a positive difference in marginal value and bid with respect to $S$.

After the screening phase, the winner selection phase begins, where the output of the screening phase 
$S$ is provided as an input.
Initially, the set of winning users $T$ is set as $S$, and 
the algorithm iterates through $T$ in the order in which users entered $S$ in the screening phase.

If the difference between the marginal value (in $T$) $v_i(T\backslash \{i\})$ and bid $b_i$ of any user $i$, given by $\sigma_i(T)$, is positive, the algorithm does the following. 
It finds a user $j \notin T$ such that $j = \arg\max_{k\in U\backslash T} v_k(T\backslash \{i\}) - b_k$.
If the difference of the marginal value obtained by replacing user $i$ with user $j$ and the bid of user $j$ is positive, i.e., $v_{j}(T\backslash \{i\}) - b_{j^*} > 0$ 
then $\gamma_i$ is set as $v_i(T\backslash \{i\}) - v_{j}(T\backslash \{i\}) + b_{j}$.
Otherwise $\gamma_i$ is set as infinity (an arbitrarily high value) and $j$ is set as $-1$.
User $j$ is a potential replacement for user $i$, and $\gamma_i$ is the critical value of the bid of user $i$ beyond which it is profitable to replace user $i$ with user $j$. 

For the critical value $\gamma_i \neq \infty$, if user $i$'s bid $b_i$ is higher than $\gamma_i$, 
then user $j$ replaces user $i$ in $T$. First the set $T$ is updated as $T\cup\{j\}\backslash\{i\}$, and  then the algorithm finds a user ${j^*} \notin T$ to determine the payment to be made to user $j$, where ${j^*} = \arg\max_{k\in U\backslash T} v_k(T\backslash \{j\}) - b_k$.
If for user ${j^*}$, $v_{{j^*}}(T\backslash \{j^*\}) - b_{{j^*}} > 0$, 
then the critical value for user $j$, $\gamma_{j} =v_{j}(T\backslash \{j\}) - v_{{j^*}}(T\backslash \{j\}) + b_{{j^*}}$.
Otherwise $\gamma_{j}= \infty$ and ${j}=-1$.
User $j$ is paid an amount equal to $\min(\gamma_{j},v_{j}(T\backslash j))$.

If $\gamma_i= \infty$ and $j=-1$, then there are no users in $U\backslash T$ that can replace user $i$ and increase platform utility.
Therefore, user $i$ is retained in $T$ and the algorithm does the following to determine user $i$'s payment.
It finds the maximum amount that user $i$ could bid and still enter the screening set $S$ and stores it as $\beta_i$.
It pays user $i$, $p_i = \min\{\beta_i, v_i(T)\}$, the minimum of $\beta_i$ and the marginal value of user $i$ in $T$.
Note that by definition $\beta_j$ is set as infinity for any user $j \notin S$ that was not selected in the screening phase of the algorithm.

If $\sigma_i(T)$ is non-positive, then the algorithm finds the user $j\in U\backslash S$ that has maximum $\sigma_j((T\backslash\{i\})\cup\{j\})$.
If $\sigma_j((T\backslash\{i\})\cup\{j\})$ is positive, then user $j$ replaces user $i$ in $T$, and the payment for user $j$ is determined as described above.
Otherwise user $i$ is removed from $T$ and its payment is set to zero.

Finally, after iterating through all elements of $T$, the bad user removal phase occurs.
Any user $i$ in $T$ having non-positive marginal utility with respect to $T$ is removed from $T$.
This ensures that only users having positive marginal utility are retained.

%An example to explain the functioning of the mechanism
\subsection{Walk Through Example}

We use the model situation presented in Fig. \ref{fig:Walk_through_example} as an example to explain the functioning of the SMART algorithm.

\begin{figure}[h]
\centering
\includegraphics{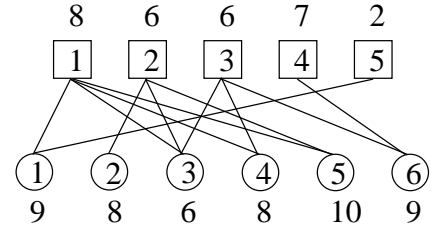}
\caption{Walk Through Example}
\label{fig:Walk_through_example}
\end{figure}

In Fig. \ref{fig:Walk_through_example}, the squares with numbers in them represent users and the circles represent tasks.
The value mentioned above a square corresponds to the bid of the user depicted by the square.
The value mentioned below the circle corresponds to the value of the task represented by that circle.
The working of the algorithm is illustrated as follows

\begin{itemize}

\item User Screening:

\begin{enumerate}
\item $S = \emptyset: v_1(\emptyset) - b_1 = 25, v_2(\emptyset) - b_2 = 18,
v_3(\emptyset) - b_3 = 17, v_4(\emptyset) - b_4 = 2, v_5(\emptyset) - b_5 = 7$.
Hence, $S \leftarrow S \cup \{1\}$.
\item $S = \{1\}: v_2(\{1\}) - b_2 = v(\{1,2\}) - v(\{1\}) - b_2 = 2$, similarly, $v_3(\{1\}) - b_3 = 3, v_4(\{1\}) - b_4 = 2, v_5(\{1\}) - b_5 = -2$.
Hence, $S \leftarrow S \cup \{3\}$.
\item $S = \{1,3\}: v_2(\{1,3\}) - b_2 = 2, v_4(\{1,3\}) - b_4 = -7, v_5(\{1,3\}) - b_5 = -2$.
Hence, $S \leftarrow S \cup \{2\}$.
\item $S = \{1,3,2\}: v_4(\{1,3,2\}) - b_4 = -7, v_5(\{1,3,2\}) - b_5 = -2$. At this point the screening phase ends and $S = \{1,3,2\}$.
\end{enumerate}

\item Winner Selection and Payment:
Initially the set of winning users $T$ is set as $S$, that is $T$ = $\{1,3,2\}$.
The algorithm iterates through $T$ with the index $i$ to determine winners and their payments.

\begin{enumerate}

\item $T(1) = \{1\}$, hence $i = 1$. Now, $v_1(T\backslash\{1\}) = v_1(T) - v_1(T\backslash\{1\}) = 9$.
The Next Best User selection phase proceeds as follows.
$U\backslash T = \{4,5\}$ and $T\backslash\{1\} = \{2,3\}$.
$v_4(\{3,2\}) - b_4 = 0 - 7 = -7$ and $v_5(\{3,2\}) - b_5 = 9 - 2 = 7$.
As $j = \arg\max_{k\in U\backslash T} v_k(\{3,2\}) - b_k$ and $v_5(\{3,2\}) - b_5 = 9 - 2 = 7 > 0$, $j = 5$ and $\beta_5 = \infty$.
$\gamma_1 = v_1(\{3,2\}) - v_5(\{3,2\}) + b_5 = 9 - 9 + 2 = 2$.
The User Entry Payment function executes as follows.
$S_1 = \emptyset$.
In the first iteration $v_2(\emptyset) - b_2 = 18$ is the $\max_{k\in U\backslash\{1\}} v_k(S_1) - b_k$, hence, $S_1 = \{2\}$.
Now, $v_1(\emptyset) - v_2(\emptyset) + b_2 = 33 - 24 + 6 = 15$.
Therefore, $\beta_1 = max(0,15) = 15$.
In the second iteration $v_3(\{2\}) - b_3 = 11$ is the $\max_{k\in U\backslash\{1,2\}} v_k(S_1) - b_k$, hence, $S_1 = \{2,3\}$.
$v_1(\{2\}) - v_3(\{2\}) + b_3 = 6$, therefore, $\beta_1 = max(15,6) = 15$.
In the third iteration $v_5(\{2,3\}) - b_5 = 7$ is the $\max_{k\in U\backslash\{1,2,3\}} v_k(S_1) - b_k$, hence, $S_1 = \{2,3,5\}$ and $\beta_1 = max(15,2) = 15$.
The iterations of the User Entry Payment function end as $v_4(\{2,3,5\}) - b_4 < 0$ and $\beta_1 = 15$ is returned. 
Now, $\sigma_1 = 9 - 8 = 1 > 0$, hence, Cond 1 is chosen. 
As $\gamma_1 = 2 < b_1 = 8$ and $\gamma_1 = 2 < \beta_1 = 15$, Cond 1.2 is chosen. 
The Replace User function in this case runs as follows.
$T = \{5,3,2\}$ and the Next Best User function yields $k = 1, \gamma_5 = 8$.
Since $\gamma_5 = 8 \neq \infty$, $p_5 = \gamma_5 = 8$. 

\item $S(2) = \{3\}$, hence $i = 3$.
Similar to iteration 1, the Next Best User selection phase yields, $j = 4$, $\beta_4 = \infty$ and $\gamma_3 = 15$.
The User Entry Payment function yields $\beta_3 = 7$.
Since, $\sigma_3(T) = 17 - 6 = 11 > 0$, Cond 1 is executed.
As $\gamma_3 = 15 > b_3 = 6$ and $\gamma_3 = 15 > \beta_3 = 7$, Cond 1.3 is followed.
The payment made is $p_3 = \min(\beta_3,v_3(T\backslash\{3\})) = \min(7,17) = 7$.

\item $S(3) = \{2\}$, hence $i = 2$.
The Next Best User function yields $j = 1$, $\gamma_2 = 16$ and $\beta_1$ is re-set to $\infty$.
The User Entry Payment function yields $\beta_2 = 8$. 
The $\sigma_2$ for this user is $v_2(T\backslash\{2\}) - b_2 = 18 - 6 = 12 > 0$.
Hence, Cond 1 is followed.
As $\gamma_2 = 16 > b_2 = 6$ and $\gamma_2 = 16 > \beta_2 = 8$, Cond 1.3 is followed.
The payment made is $p_2 = 8$.
\end{enumerate}

\item Bad User Removal:
$v_2(\{3,5\}) - p_2 = 18 - 8 = 10 > 0$.
$v_3(\{2,5\}) - p_3 = 17 - 7 = 10 > 0$.
$v_5(\{2,3\}) - p_5 = 9 - 8 = 1 > 0$.
Therefore no user is removed from $T$.

\item Result of the Algorithm: The final set of winning users, $T = \{2,3,5\}$ and the payment made to these users is $p_2 = 8, p_3 = 7$ and $p_5 = 8$.
Therefore, the value of the tasks performed by the users is $v(T) = 50$ and the payment made to these users is $\sum_{i\in T} p_i = 23$.
Hence the platform utility in this case is
\begin{equation*}
u(T) = v(T) - \sum_{i\in T} p_i,
\end{equation*}

\begin{equation*}
u(T) = 27.
\end{equation*}

\item For this example, M-Sensing \cite{M-Sensing2012} produces a platform utility of 20.
The set of screened users $S$ of SMART are the set of winning users of M-Sensing $\{1,3,2\}$.
M-Sensing pays each winning user $i$ an amount equal to $\beta_i$ calculated from the User Entry Payment function, therefore, $p_1 = \beta_1 = 15, p_2 = \beta_2 = 8, p_3 = \beta_3 = 8$. 
\end{itemize}

The next section shows that SMART satisfies the four properties of an \emph{efficient} auction.

%Description and proof of the mechanism's properties
\subsection{Properties of the Algorithm}

\begin{lem}\label{computational_efficiency_smart}
SMART is computationally efficient.
\end{lem}

\begin{proof} User Screening phase takes $O(nm^2)$ time:  
The screening phase selects at most $m$ users since we have at most $m$ tasks, and each step to select the best user with largest marginal utility takes $O(nm)$ time.

Winner Selection phase takes $O(nm^3)$ time:  Again, since initially the size of set $|T|=|S| \le m$, the main routine in the Winner Selection phase that runs through the elements of $T$ takes at most $O(m)$ time. There are four sub-routines in the Winner Selection phase: Next Best User Function, Replace User Function, User Entry Payment and Bad User Removal.
The Next Best User function takes at most $O(n)$ time since it finds the possible replacement among the $n$ users. 
Since Replace User function invokes the Next Best User function, it also takes at most $O(n)$ time.
The User Entry Payment function is similar to the User screening routine and takes $O(nm^2)$ time.
The Bad User Removal phase also sequentially runs for $|T|\le m$ time, hence it takes $O(m)$ time.
Hence, the computational complexity of the SMART algorithm is at most $O(nm^3)$, hence polynomial. 
Therefore, SMART is computationally efficient.
\end{proof}

\begin{lem}\label{lem:offrational}
SMART is individually rational.
\end{lem}

\begin{proof} Since user $i$ bids $b_i$ and user $i$ is assumed to be rational, $b_i - c_i > 0$.
Therefore, to prove individual rationality of the SMART algorithm, we only need to prove that
$\forall \ i\  \in T$, $p_i \geq b_i$, where $T$ is the final set of winning users.
In the SMART Algorithm, the payment made to winning users depends on either Cond 1.1, or Cond 1.2, or Cond 1.3, or Cond 2. We discuss each case individually. 

\begin{itemize}
\item  Cond 1.1. User $i$ is paid according to Cond 1.1 if and only if $\gamma_i \ge b_i$, where $\gamma_i$ is the critical bid value above which it is profitable to replace user $i$ with another user from $U\backslash T$.
Since user $i$ is paid $p_i = \gamma_i$, hence, $p_i \ge b_i$.

\item  Cond 1.2. In this case $\gamma_i < b_i$ and $\gamma_i \neq \infty$, i.e. there exits an user $j\in U\backslash T$ that will replace user $i\in T$, where user $j$ is such that   
 $j = \arg\max_{k\in U\backslash T} v_k(T\backslash \{j\}) - b_k$, computed at Line 14 of the SMART (by running the Next Best User function). This is the most involved cases to prove, where we have to show that the payment made to a replacement user $j$ is more than or equal to its bid.
 
Let $T$ before the replacement be $T_{\text{old}}$, and after the replacement be $T_{\text{new}} = T\cup\{j\}\backslash\{i\}$. 
Note that, since $\gamma_i \neq \infty$ that is computed at Line 14 of the SMART (by running the Next Best User function), it implies that 
\begin{equation}\label{eq:rational1}
v_j(T_{\text{old}}\backslash \{i\}) - b_j > 0,
\end{equation}
from Line 2 of the Next Best User function.

To find the payment for the replacement user $j$, $p_j$, at Line 24 of SMART, first the Replace User function is called with input $\{T_{\text{old}},i,j\}$, that in turn calls the Next Best User function with input 
$\{T_{\text{old}},j\}$. The Next Best User function finds user 
 $j^{*} = \arg\max_{k\in U\backslash T_{\text{new}}} v_k(T_{\text{new}}\backslash \{j\}) - b_k$.
 Since $T_{\text{new}} = T\cup\{j\}\backslash\{i\}$, $j^* \in  U\backslash T_{\text{old}}$ or $j^*=i$. 
 
 \begin{itemize}
\item 
In the Next Best User Function if $v_{j^{*}}(T_{\text{new}}\backslash \{j\}) - b_{j^{*}} > 0$, then the critical value for user $j$ is 
 $\gamma_{j} = v_j(T_{\text{new}}\backslash \{j\}) -  v_{j^*}(T_{\text{new}}\backslash \{j\}) + b_{j^*}$. Then returning to the  Replace User function to find $p_j$, if $\gamma_j \le  v_j(T_{\text{new}}\backslash \{j\})$, then $p_j=\gamma_j$, otherwise $p_j=v_j(T_{\text{new}}\backslash \{j\})$. 
Note that from (\ref{eq:rational1}), we know that $v_j(T_{\text{new}}\backslash \{j\}) \ge b_j$. 
Since $T_{new}\backslash \{j\} = T_{old}\backslash \{i\}$ we only need to prove that whenever $p_j=\gamma_j$, $\gamma_j \ge b_j$.
 
If $j^* \in  U\backslash T_{\text{old}}$, then since user $j$ was the best replacement for user $i$ among users in $T_{\text{old}}\backslash \{i\} = T_{\text{new}}\backslash \{j\}$ 
 computed at line 14 of SMART, we have that $v_j(T_{\text{new}}\backslash \{j\}) - v_{j^*}(T_{\text{new}}\backslash \{j\}) + b_{j^*} - b_j \ge 0$. Thus, $\gamma_j \ge b_j$. 
 
Otherwise, if $j^*=i$, then from Line 3 of the Next Best User function, 
$\gamma_j = v_j(T_{\text{new}}\backslash \{j\}) - v_i(T_{\text{new}}\backslash \{j\}) + b_i$, which is also equivalent to $\gamma_j = v_j(T_{\text{old}}\backslash \{i\}) - v_i(T_{\text{old}}\backslash \{i\}) + b_i$. Since with Cond 1.2, $\gamma_i < b_i$, where $\gamma_i = v_i(T_{\text{old}}\backslash \{i\}) - v_j(T_{\text{old}}\backslash \{i\}) + b_j$.
$v_i(T_{\text{old}}\backslash \{i\}) - v_j(T_{\text{old}}\backslash \{i\}) + b_j - b_i < 0$.
Therefore, $\gamma_j > b_j$. 

\item If $v_{j^{*}}(T_{\text{new}}\backslash \{j\}) - b_{j^{*}} \le 0$, then $\gamma_j=\infty$ and $p_j = v_j(T_{\text{new}}\backslash \{j\})$. The claim follows from (\ref{eq:rational1}), where it is shown that $v_j(T_{\text{new}}\backslash \{j\}) \ge  b_j$.
  
\end{itemize}

\item  Cond 1.3. Here, $p_i = \min(\sigma_i(T) + b_i,\beta_i)$.
Let $p_i = \sigma_i(T) + b_i$. As $\sigma_i(T) = v_i(T\backslash\{i\}) - b_i$, $p_i = v_i(T\backslash\{i\})$.
Since Cond 1.3 is a Sub-Condition of Cond 1, which requires that $v_i(T\backslash\{i\}) > b_i$, $p_i > b_i$.
Let $p_i = \beta_i$. Note that
$\beta_i$ represents the maximum bid that user $i$ could have bid and still enter $S$. 
Since, user $i$ was selected, $b_i < \beta_i$ which implies $b_i < p_i$.

\item Cond 2. The situation becomes equivalent to Cond 1.2.
As the payment made according to Cond 1.2 has been proved to be individually rational, so is the payment made according to Cond 2.
\end{itemize}
\end{proof}

We next show that the proposed SMART algorithm is profitable.
\begin{lem}\label{Profitability_Lemma}
SMART is profitable.
\end{lem}

\begin{proof} To prove this Lemma, we take an indirect route by showing that the utility of SMART is at least as much   as M-Sensing (Lemma \ref{SMART_more_profitable_than_M-Sensing}), and the result follows since M-Sensing \cite{M-Sensing2012} is profitable. \end{proof} 
This gives us one short proof for showing two results, that SMART is profitable and is at least as profitable as M-Sensing.

\begin{lem}\label{SMART_more_profitable_than_M-Sensing}
The utility of SMART is greater than or equal to the utility of M-Sensing \cite{M-Sensing2012}.
\end{lem}
\begin{proof}
The output set $S$ of the screening phase of SMART is equal to the final selection set of M-Sensing.
SMART then refines $S$ further through the Winner Selection phase and the Bad User removal phase. 
In the winner selection phase, SMART initially sets the set of winning users $T$ to be equal to $S$ and then iteratively updates $T$. Moreover, note that with M-Sensing, and SMART, the payment for each user $i \in S$, is $\beta_i$, and $p_i = \min\{\gamma_i, \sigma_i(T)+b_i, v_i\beta_i\}$, respectively. Thus, $p_i\le \beta_i$. 

So initially, 
\begin{equation}\label{Initial_SMART_M-Sensing}
v(S) - \sum_{l\in S} \beta_l = v(T) - \sum_{l\in T} \beta_l,
\end{equation}
since $T=S$.

In each iteration of the winner selection phase of SMART the current user, say user $i$ is either retained, removed or replaced.

\begin{itemize}

\item 
If the current user $i$ is {\it retained} there is no change in $T$, and still $T=S$, however, since $p_i \leq \beta_i$,

\begin{equation}\label{UserRetained}
v(S) - \sum_{l\in S} \beta_l \leq v(T) - p_i - \sum_{l\in S\backslash \{i\}} \beta_l.
\end{equation}

\item
The current user $i$ is {\it removed} if and only if $\sigma_i(T) - b_i \leq 0$.
Therefore, $v_i(T\backslash \{i\}) - b_i \leq 0$, which implies 
$v(T) - v(T\backslash \{i\}) - b_i \leq 0$, and finally, $v(T\backslash \{i\}) \geq v(T) - b_i$.

Since M-Sensing is rational, $\beta_i \geq b_i$, we have  $v(T\backslash \{i\}) \geq v(T) - \beta_i$.

Therefore, after removing user $i$, with $T \leftarrow T \backslash \{i\}$,
\begin{equation}\label{UserRemoved}
v(S) - \sum_{l\in S} \beta_l \leq v(T) - \sum_{l\in T} \beta_l.
\end{equation}

\item
The {\it replacement} of the current user $i$ can occur at two points in the SMART algorithm, in Cond 1.2 at lines 23, 24 and in Cond 2 at lines 30, 31.
In both cases, the Replace User function is called to replace user $i$ with another user say user $j\in U\backslash T$.
Let the set of winning users be $T_{\text{old}}$ before replacement and $T_{\text{new}}$ after replacement.
$T_{\text{new}} = (T_{\text{old}} \backslash \{i\}) \cup \{j\}$.

\begin{itemize}

\item
Consider the replacement done in lines 23 and 24.
In this situation $\gamma_i \leq b_i$, $\gamma_i \leq \beta_i$ and $\gamma_i \neq \infty$.
The Replace User function first performs replacement and then calls the Next Best User function with inputs as $\{j,U,T_{\text{new}}\}$.
The Next Best User function finds $k = \arg\max_{m\in U\backslash T_{\text{new}}} v_m(T_{\text{new}}\backslash \{j\}) - b_m$ and returns $\gamma_j = v_j(T_{\text{new}}\backslash\{j\}) - v_k(T_{\text{new}}\backslash\{j\}) + b_k$.
Since $i \in U\backslash T_{\text{new}}$, therefore, $v_k(T_{\text{new}}\backslash \{j\}) - b_k \geq v_i(T_{\text{new}}\backslash \{j\}) - b_i$.
This implies that in the worst case, the replacement user $k$ is  same as user $i$.
Therefore, in the worst case $\gamma_j = v_j(T_{\text{new}}\backslash\{j\}) - v_i(T_{\text{new}}\backslash\{j\}) + b_i$.
Since the payment made in the Replace User function is $p_j = \min(\gamma_j,v_j(T_{\text{new}}\backslash\{j\}))$, $p_j \leq \gamma_j$.
Therefore,

$p_j \leq v_j(T_{\text{new}}\backslash\{j\}) - v_i(T_{\text{new}}\backslash\{j\}) + b_i$, which on rearranging, 
$v_i(T_{\text{new}}\backslash\{j\}) - b_i \leq v_j(T_{\text{new}}\backslash\{j\}) - p_j$.
Since $T_{\text{new}}\backslash \{j\} = T_{\text{old}}\backslash \{i\}$, it follows that 
$v_i(T_{\text{old}}\backslash\{i\}) - b_i \leq v_j(T_{\text{new}}\backslash\{j\}) - p_j$.

As M-Sensing is individually rational, $\beta_i \geq b_i$, therefore, 
$v_i(T_{\text{old}}\backslash\{i\}) - \beta_i \leq v_j(T_{\text{new}}\backslash\{j\}) - p_j$.

\item
Consider the replacement done in lines 30 and 31.
In this situation $\sigma_i(T_{\text{old}}) \leq 0$ and $\gamma_i \neq \infty$.
Since, $\gamma_i \neq \infty$, $v_j(T_{\text{old}} \backslash \{i\}) - b_j \geq 0$.
Since, $\sigma_i(T_{\text{old}}) \leq 0$, $v_i(T_{\text{old}}\backslash\{i\}) - b_i \leq 0$.
As M-Sensing is individually rational, $\beta_i \geq b_i$ therefore,
$v_i(T_{\text{old}}\backslash\{i\}) - \beta_i \leq 0$.

Also, since the payment made in the replace user function $p_j = \min(\gamma_j,v_j(T_{\text{new}}\backslash\{j\}))$, $p_j \leq v_j(T_{\text{new}}\backslash\{j\})$.
Hence, 
$
v_j(T_{\text{new}}\backslash\{j\}) - p_j \geq 0$.

Therefore, $v_i(T_{\text{old}}\backslash\{i\}) - \beta_i \leq v_j(T_{\text{new}}\backslash\{j\}) - p_j$.

\end{itemize}

Consider the first iteration of SMART in the winner selection phase for replacement.
Here, $S = T_{\text{old}}$ and $T_{\text{new}}\backslash \{j\} = S \backslash \{i\}$.
Therefore,
$v_i(S\backslash\{i\}) - \beta_i \leq v_j(T_{\text{new}}\backslash\{j\}) - p_j$, which implies 
$v(S) - v(S\backslash\{i\}) - \beta_i \leq v(T_{\text{new}}) - v(T_{\text{new}}\backslash\{j\}) - p_j$.
Since $v(T_{\text{new}}\backslash\{j\}) = v(S\backslash\{i\})$,
$v(S) - \beta_i \leq v(T_{\text{new}}) - p_j$.
Subtracting $\sum_{l\in S\backslash\{i\}} \beta_l$ to both sides, we get
\begin{equation}\label{ReplacementFirstIteration}
v(S) - \sum_{l\in S} \beta_l \leq v(T_{\text{new}}) - p_j - \sum_{l\in S\backslash\{i\}} \beta_l.
\end{equation}

Now consider the situation when $n$th replacement occurs.
Let the set of winning users before replacement be $T_{\text{old}_n}$ and let the set of winning users after replacement be $T_{\text{new}_n}$.
Therefore,

\begin{equation}\label{T_old_n}
T_{\text{old}_n} = (S\backslash\{i_1,\ldots,i_{n-1}\})\cup\{j_1,\ldots,j_{n-1}\}),
\end{equation}

\begin{equation}\label{T_new_n}
T_{\text{new}_n} = (S\backslash\{i_1,\ldots,i_n\})\cup\{j_1,\ldots,j_n\}).
\end{equation}

Generalizing  \eqref{ReplacementFirstIteration} at the $n$th iteration, we get, $v(S) - \sum_{l\in S} \beta_l$
\begin{equation}\label{Replacement_N_Iteration}
 \leq v(T_{\text{new}_n}) - \sum_{m\in \{j_1,\ldots,j_n\}} p_j - \sum_{l\in S\backslash\{i_1,\ldots,i_n\}} \beta_l.
\end{equation}

Let \eqref{Replacement_N_Iteration} hold true.
Lets analyze the $n+1^{th}$ iteration.
Using the same notation, 

\begin{equation}\label{T_old_n+1}
T_{\text{old}_{n+1}} = (S\backslash\{i_1,\ldots,i_n\})\cup\{j_1,\ldots,j_n\}),
\end{equation}

\begin{equation}\label{T_new_n+1}
T_{\text{new}_{n+1}} = (S\backslash\{i_1,\ldots,i_{n+1}\})\cup\{j_1,\ldots,j_{n+1}\}).
\end{equation}

At the $n+1$th replacement, 
\begin{equation}
\nonumber v_{i_{n+1}}(T_{\text{old}_{n+1}}\backslash\{i_{n+1}\}) - \beta_i \leq v_{j_{n+1}}(T_{\text{new}_{n+1}}\backslash\{j_{n+1}\}) - p_j.
\end{equation}

Therefore, $v(T_{\text{old}_{n+1}}) - v(T_{\text{old}_{n+1}}\backslash\{i_{n+1}\}) - \beta_i \leq v(T_{\text{new}_{n+1}}) - v(T_{\text{new}_{n+1}}\backslash\{j_{n+1}\}) - p_j$.
Since, $T_{\text{old}_{n+1}}\backslash\{i_{n+1}\} = T_{\text{new}_{n+1}}\backslash\{j_{n+1}\}$,

\begin{equation}
\nonumber v(T_{\text{old}_{n+1}}) - \beta_i \leq v(T_{\text{new}_{n+1}}) - p_j.
\end{equation}

From \eqref{T_new_n} and \eqref{T_old_n+1}, we note that $T_{\text{old}_{n+1}} = T_{\text{new}_n}$.
Substituting for $T_{\text{new}_n}$ from \eqref{Replacement_N_Iteration}, we get, 
$v(S) - \sum_{l\in S} \beta_l - \beta_{i_{n+1}} \leq v(T_{\text{new}_{n+1}}) - p_{j_{n+1}} - \sum_{m\in \{j_1,\ldots,j_n\}} p_j - \sum_{l\in S\backslash\{i_1,\ldots,i_n\}} \beta_l$.
Hence, $v(S) - \sum_{l\in S} \beta_l$
\begin{equation}\label{Replacement_N+1_Iteration}
 \leq v(T_{\text{new}_{n+1}}) - \sum_{m\in \{j_1,\ldots,j_{n+1}\}} p_j - \sum_{l\in S\backslash\{i_1,\ldots,i_{n+1}\}} \beta_l.
\end{equation}

Therefore, from  \eqref{ReplacementFirstIteration}, \eqref{Replacement_N_Iteration}, and \eqref{Replacement_N+1_Iteration}, using  induction we get that, $v(S) - \sum_{l\in S} \beta_l$
\begin{equation}\label{UserReplaced}
 \leq v(T_{\text{new}_n}) - \sum_{m\in \{j_1,\ldots,j_n\}} p_j - \sum_{l\in S\backslash\{i_1,\ldots,i_n\}} \beta_l,
\end{equation}

is true for any arbitrary $n$.

\end{itemize}

In the winner selection phase, SMART iterates through all the elements of $S$ in a sequential manner.
Using  \eqref{UserRetained}, \eqref{UserRemoved}, and, \eqref{UserReplaced}, we can conclude that at the end of the winner selection phase,
\begin{equation}\label{SMARTM-Sensing}
v(S) - \sum_{l\in S} \beta_l \leq v(T) - \sum_{l\in T} p_l,
\end{equation}
where the LHS  represents the utility of M-Sensing and the RHS represents the utility of SMART, since the Bad User Removal function does not decrease the utility.

\end{proof}

\begin{remark}
In terms of profitability, M-Sensing performs the same as SMART, when SMART retains all the users obtained in the Screening phase and pays each user $i$ an amount $\beta_i$.
This happens when all of the following conditions are satisfied for each user $i\in S$, where $S$ is the set of screened users obtained after the screening phase of SMART.
\begin{itemize}
\item $v_i(S\backslash\{i\}) > 0$.
\item The value of $\gamma_i$ as computed by the Next Best User function at line 14 of SMART, is $\infty$ in which case no valid replacement user exists or when $\gamma_i > \beta_i$.
\item $\beta_i < v_i(S\backslash\{i\})$.
\end{itemize}
In all other scenarios, SMART has larger utility than M-Sensing.
\end{remark}

Next, we show the most important property of SMART, its truthfulness. Towards that end we will use the  Myerson's Theorem \cite{Myerson1981Auction}.
\begin{thm}\cite{Myerson1981Auction}\label{Myerson_Theorem}
A reverse auction is considered truthful if and only if
\begin{itemize}
\item The selection rule is monotone. If a user $i$ wins the auction by bidding $b_i$, it would also win the auction by bidding an amount $b_i'$, where $b_i' < b_i$.
\item Each winner is paid a critical amount. If a winning user submits a bid greater than this critical value, it will not get selected.
\end{itemize}
\end{thm} 

\begin{lem}
SMART is truthful.
\end{lem}

\begin{proof}
\emph{ Monotonicity of SMART:}
Consider a user $i$ that is selected by SMART with bid $b_i$, i.e., $i \in T$.
Let user $i$ change its bid to $b_i'$, where $b_i' < b_i$. 
Let $\sigma_i'(T) = v_i(T\backslash \{i\}) - b_i'$. Now we look at two cases where the user $i$ with bid $b_i$ could have entered $T$.

a) Assume that user $i\in S$, in the user screening phase, with bid $b_i$, and was then retained in $T$. 
By the definition of the User Screening phase, if $i\in S$, then $i =\arg \max_{k\in U\backslash S} v_k(S) - b_k$ at some iteration $r_i$, where user $i$ entered $S$. 
If user $i$ instead bid $b_i' < b_i$, then again user $i$ enters $S$ at iteration $r_i$ or earlier. 
Since user $i$ is retained in $T$ with bid $b_i$, by definition of SMART, we have $\sigma_i(T) > 0$ and $\gamma_i > b_i$. 
Consequently, $\sigma_i'(T) > 0$ and $\gamma_i > b_i'$. 
Thus, even if user $i$ bid $b_i'$, both Cond. 1 and Cond. 1.1 in the Winner Selection phase are satisfied, and therefore user $i$ is retained in $T$.

b) If user $i$ entered into $T$ by replacing some user $j \in T$, then user $i$ has to satisfy  
$i = \arg\max_{k\in U\backslash T} v_k(T\backslash\{j\}) - b_k$ in the Next Best User function when called from Line 14 of the SMART.
Therefore, should user $i$ decrease its bid to $b_i'$ it would still replace user $j \in T$ and enter $T$. 

\emph{Existence of a Critical Bid Amount with SMART:} We claim that the payment $p_i$ made by SMART is critical, i.e., if any user $i$ bids in excess of its critical amount $p_i$, then SMART will not select it.
Two possible conditions exist with SMART, 

\begin{itemize}

\item A winning user $i \in S, i \in T$ receives $p_i= \min(v_i(T\backslash\{i\}),\beta_i,\gamma_i)$, where $S$ is the set of screened users. 
Lets assume that $p_i = v_i(T\backslash \{i\})$, and if user $i$ changes its bid to $b_i'$, $b_i' > v_i(T\backslash\{i\})$, then $\sigma_i'(T) = v_i(T\backslash\{i\}) - b_i' < 0$. 
From Cond 2 of SMART, we conclude that user $i$ would be replaced by another user $j$ if possible in $T$ or removed from $T$. 
If $p_i=\beta_i$ and $b_i' > \beta_i$, then by the definition of the User Entry Payment function, user $i$ will not enter the screening set $S$ in the User Screening phase.
Therefore, user $i$ will not enter $T$.
Finally, if $p_i = \gamma_i$ and $b_i' > \gamma_i$, then from Cond 1.2 user $i$ will be replaced by some user $j$ in $U\backslash T$.

\item A winning user $i\in T$ and $i \notin S$ receives $p_i= \min(v_i(T\backslash\{i\}),\gamma_i)$.
Such a user does not belong to $T$ at the beginning of the Winner Selection Phase and is a replacement user.
Note that, replacement users are paid inside the Replace User Function.
In the function a replacement user say $j \notin S$ replaces a user $i \in T$.
Let $T$ before the replacement be $T_{\text{old}}$, and after the replacement be $T_{\text{new}} = T\cup\{j\}\backslash\{i\}$. 
The Replace User function calls the Next Best User function to compute $\gamma_j$ for user $j$.
This function returns both the second best user to user $j$, user $k$ and the value of $\gamma_j$.
  
Lets assume that user $j$ is paid an amount equal to $\gamma_j$.
From the Next Best User Function $\gamma_j = v_j(T_{\text{new}}\backslash\{j\}) - v_k(T_{\text{new}}\backslash\{j\}) + b_k$. 
If user $j$ bids an amount $b_j' > \gamma_j$, then the following happens.
Since $T_{\text{new}}\backslash\{j\} = T_{\text{old}}\backslash\{i\}$ from line 1 of the Next Best User function we can conclude that user $k$ would take user $j$'s place as a replacement user at line 14 of SMART.    

Consider the alternate case when user $j$ is paid an amount equal to $v_j(T_{\text{new}}\backslash \{j\})$.
Note that user $j$ was selected as a replacement for user $i$, through the Next Best User function call made in line 14 of SMART for the computation of $\gamma_i$ for user $i$. At this point the set of winning users is $T_{\text{old}}$. 
If user $j$ bids $b_j' > v_j(T_{\text{new}}\backslash\{j\})$, then since $T_{\text{new}}\backslash\{j\} = T_{\text{old}}\backslash\{i\}$, $v_j(T_{\text{old}}\backslash\{i\}) = v_j(T_{\text{new}}\backslash\{j\})$, then 
from line 2 of the Next Best User function called at line 14 of SMART, since $b_j' > v_j(T_{\text{new}}\backslash\{j\})$, the Next Best User function either returns some other user $k$ as a replacement for user $i$ or if the user $j$ is the maximizer in line 1 of  the Next Best User function, then $j=-1$ is returned. In conclusion, user $j$ does not replace user $i$ if $b_j' > v_j(T_{\text{new}}\backslash\{j\})$.
\end{itemize}

\end{proof}

\subsection{Discussion} The SMART algorithm proposed in this section for maximizing the utility of the platform for offline crowd-sourcing problem is a simple to implement algorithm inspired by the VCG mechanism (that is known to be truthful but extremely hard to compute). With SMART, we first shortlist the potential winners using a greedy (linear time) algorithm that at each step finds the user with the best marginal utility. Then each user in the shortlist is kept or dropped or replaced depending on the effective utility that user brings to the platform. The payment strategy of SMART is similar to VCG mechanism, where each user is paid for the marginal utility it brings to the platform, and is equal to the increase in utility by replacing any user by the next best user plus the bid of the next best user, unless for some exceptional cases where it differs slightly. 

In comparison to the earlier algorithm M-Sensing \cite{M-Sensing2012}, SMART is more selective in picking the winning users and more frugal in payment made to each selected user. With M-Sensing, each winning user is paid an amount equal to the maximum amount that any user can bid and still be selected. We show that M-Sensing pays a little too much, since SMART is more profitable than M-Sensing while being truthful. 

\section{Online Scenario}

Unlike the offline scenario, where all the users arrive together, in the online scenario the users arrive sequentially. 
The challenge posed here is that the platform must respond to each incoming user immediately and irrevocably, though it may not have any prior knowledge of the bidding profiles of the users coming in the future.
We adapt the SMART algorithm for the online scenario and ensure that the resulting algorithm adheres to the four properties of an efficient mechanism.
In order to facilitate the analysis of the problem, we make the following assumptions in the online scenario.
\begin{itemize}
\item Without loss of generality, we assume that the every user in $U$, is such that the total value of tasks it provides is greater than its bid 
i.e., $v_i>b_i$.\footnote{If an incoming user has $v_i<b_i$, then it is rejected immediately and does not count as a user in $U$.}  
\item The user does not know its time of arrival with reference to the time of arrival of the other users.
\end{itemize}
We define $U[p:q]$ as the set representing the successive users from time $p$ to $q$.

\begin{table}
\begin{tabular}{r l}
\hline
& \textbf{ONLINE-SMART}\\
\hline
1 	&	\textbf{// Initialization}\\
2 	&	$T \leftarrow \emptyset$, $R \leftarrow \emptyset$\\ 
3	&	$P \leftrightarrow \{p_1,p_2,\ldots p_n\}$ \\
4	&	$n = |U|$, $k = \lfloor n/c \rfloor$\\
5 	&	\textbf{// Phase 1 - Observation}\\
6       &       $R \leftarrow \text{SMART}(U[1:k])$ //R is the set of winning users of SMART\\
7      	&       $p_1,p_2,\ldots p_n \leftarrow 0$ \\
8	&       \textbf{// Phase 2 - Winner Selection}\\
9	&	\textbf{for} each $i = k+1,k+2,\ldots,n $ and $|T| \leq m$ do\\
10	&	\quad \textbf{if} $|R| < m$ then\\
11	&	\quad \quad \textbf{if} $v_i(R) - b_i > 0$ then\\
12	&	\quad \quad \quad $(R,T,P) \leftarrow$ Add User $(R,i)$\\ 
13	&	\quad \quad \textbf{else}\\
14	&	\quad \quad \quad $(R,T,P) \leftarrow$ Try To Replace $(R,i)$\\
15	&	\quad \quad \textbf{end if}\\
16	&	\quad \textbf{else if} $|R| = m$\\
17	&	\quad \quad $(R,T,P) \leftarrow$ Try To Replace $(R,i)$\\
18	&	\quad \textbf{end if}\\
19	&	\quad $R \leftarrow$ Remove Bad Users $(R,T)$\\ 
20	&	\textbf{endfor}\\
21	&	Return $(T,P)$\\
\hline
\end{tabular}
\label{online_smart} 
\end{table}

\begin{table}
\begin{tabular}{r l}
\hline
& \textbf{Add User $(R,i)$}\\
\hline
1	&	$p_i \leftarrow v_i(R)$\\
2	&	$R \leftarrow R \cup \{i\}$, $T \leftarrow T \cup \{i\}$\\
3	&	Return $(R,T,P)$\\
\hline
\end{tabular} 
\end{table}

\begin{table}
\begin{tabular}{r l}
\hline
& \textbf{Try To Replace $(R,i)$}\\
\hline
1	&	$j = \arg\max_{k\in R\backslash T} \Big(v(R\backslash \{k\}) \cup \{i\}) - v(R) + b_k - b_i \Big)$\\
2	&	\textbf{if} $v((R\backslash \{j\}) \cup \{i\}) - v(R) + b_j > b_i$ then\\
3	&	\quad $p_i \leftarrow v((R\backslash \{j\}) \cup \{i\}) - v(R) + b_j$\\
4	&	\quad $R \leftarrow (R \backslash \{j\}) \cup \{i\}$ and $T \leftarrow T \cup \{i\}$\\
5	&	\textbf{end if}\\
6	&	Return $(R,T,P)$\\
\hline
\end{tabular} 
\end{table}

\begin{table}
\begin{tabular}{r l}
\hline
& \textbf{Remove Bad Reference Users $(R,T)$}\\
\hline
1       &	\textbf{for} each $j \in R\backslash T$ do\\
2	&	\quad \textbf{if} $v(R) < v(R\backslash \{j\}) + b_j$ then\\
3	& 	\quad \quad $R \leftarrow R \backslash \{j\}$\\
4	& 	\quad \textbf{end if}\\
5	& 	\textbf{endfor}\\
6	&	Return $R$\\
\hline
\end{tabular} 
\end{table}

%This section discusses the proposed algorithm
\subsection{ONLINE-SMART Algorithm}

The ONLINE-SMART algorithm, motivated by the $k$-secretary problem, consists of two phases, the \textit{observation phase} and the \textit{winner selection phase}.
In the observation phase, it rejects the first $k = \lfloor \frac{n}{c} \rfloor$ (c is a constant chosen by the platform) users on their arrival.
It then runs SMART on the bidding profiles of $U[1:k]$ and stores the output set of winning users as a reference set $R$. 
Since the number of tasks is $m$, the cardinality of $R$ is less than or equal to $m$. 
The algorithm uses this set $R$ as a reference for selection of users among the remaining $n-k$ users, $U[n-k:n]$.
In the beginning of the selection phase, the final winners set $T$ is set to $\emptyset$ and users are processed as they arrive. 

On the arrival of a user $i\in U[k+1:n]$, if the cardinality of $R$ is less than $m$ then the algorithm does the following.
If $v_i(R) - b_i > 0$, that is the difference in marginal value and bid with respect to $R$ is positive, then the algorithm adds user $i$ to both $R$ and $T$ and makes a payment $p_{i}=v_{i}(R)$.
If $v_i(R) - b_i \leq 0$ then the algorithm calls the Try To Replace function.
The Try To Replace function determines if it is profitable to replace a user $j\in R$ with user $i$ (if $v((R\backslash \{j\}) \cup \{i\}) - v(R) + b_j > b_i$). 
If so the algorithm replaces user $j$ with user $i$.
It makes a payment to user $i$, $p_{i}=v_i((R\backslash \{j\}) \cup \{i\}) - v(R) + b_{j}$. 
If no user $j\in R$ can be replaced profitably, user $i$ is rejected and the algorithm moves on to the next user.

If the cardinality of $R$ is equal to $m$ then there can be no further addition of users to $R$ without decreasing the marginal utility of $R$.
This is because the presence of $m$ users implies the completion of at least $m$ tasks which leaves no more new tasks to be done.
Therefore, the algorithm calls the Try To Replace function to find if any user $j\in R$ can be replaced by the current user, whose execution has been explained before.
If no user can be replaced profitably, user $i$ is rejected and the algorithm moves on to the next user.

Between the arrivals of two users in the selection phase the algorithm iterates through $R\backslash T$ and removes any user $i$ having $v_i(R\backslash\{i\}) - b_i \leq 0$.
This ensures that all users in $R$ have a positive difference in marginal value and bid, and consequently form a "good" reference for the incoming users.
The algorithm terminates when either all users in $U$ have arrived or when $|T|$ becomes equal to $m$. 

%The properties of ONLINE-SMART are discussed here    
\subsection{Properties of the Mechanism}

\begin{lem} 
ONLINE-SMART is computationally efficient.
\end{lem}
\begin{proof} In the observation phase, the platform runs the SMART algorithm on the set of the first $k$ users.
Therefore, from Lemma \ref{computational_efficiency_smart} the complexity associated with observation phase is $O(nm^3)$ (since  $k = \frac{n}{c}$). 
In the selection phase, the main routine runs $n-k$ times.
Within the main routine there are two functions, one each for the Try To Replace and the Remove Bad Reference Users.
Both of these functions have a complexity of $O(m^2)$.
Consequently, the selection phase takes $O(nm^2)$ time. 
Therefore the computational complexity associated with ONLINE-SMART is $O(nm^3)$. 
Hence, ONLINE-SMART is computationally efficient.\end{proof}

\begin{lem} 
ONLINE-SMART is individually rational.
\end{lem}
\begin{proof} To prove that ONLINE-SMART is individually rational, it is sufficient to prove that $\forall_{i\in T}$ $p_i\geq b_i$. 
In the selection phase of ONLINE-SMART, a user enters the final winner's set $T$ through either the Add User function or the Try To Replace function.
The Add User function is called when the incoming user $i$ has $v_i(R) > b_i$.
Since the Add User function pays user $i$, $p_i = v_i(R)$, $p_i > b_i$.
An incoming user $i$ replaces an existing user $j\in R\backslash T$ through the Try To Replace function if and only if $v((R\backslash \{j\}) \cup \{i\}) - v(R) + b_j - b_i > 0$.
Since the payment made in the Try To Replace function to such a selected user $i$ is $p_i = v((R\backslash \{j\}) \cup \{i\}) - v(R) + b_j$, therefore $p_i > b_i$.\end{proof}

\begin{lem} 
ONLINE-SMART is profitable.
\end{lem}
\begin{proof} 
An incoming user $i\in U[k+1:n]$ enters $T$ through either the Add User function or the Try To Replace function.

Let user $i$ enter $T$ through the Add User function. 
At any point in the algorithm the set of winning users ($T$) is always a subset of the reference set ($R$), that is $T\subseteq R$.
This implies that for a user $i\in U\backslash R$, the marginal tasks with respect to $R$ are always a subset of the marginal tasks with respect to $T$, i.e. $\tau_i(R) \subseteq \tau_i(T)$.
Therefore, $\chi(\tau_i(R)) \leq \chi(\tau_i(T))$.
The payment made to user $i$ is $p_i = v_i(R) = \chi(\tau_i(R))$.
Hence, $v_i(T) - p_i = \chi(\tau_i(T)) - \chi(\tau_i(R)) \geq 0$. 
The platform utility of $T\cup\{i\}$ is $u(T\cup\{i\}) = v(T) + v_i(T) - p_i - \sum_{j\in T} p_i = u(T) + v_i(T) - p_i$.
Therefore $u_i(T) = u(T\cup\{i\}) - u(T) \geq 0$, i.e. the incremental utility change is always non-negative. 

Let user $i$ enter $T$ through the Try To Replace function.
The payment made to user $i$ is 
\begin{eqnarray}
\nonumber p_i &=& v((R\backslash \{j\}) \cup \{i\}) - v(R) + b_{j},\\ 
\nonumber     &=& v((R\backslash \{j\}) \cup \{i\}) - v(R\backslash\{j\}) + v(R\backslash\{j\}) - v(R) + b_j,\\
\nonumber     &=& \chi(\tau_i(R\backslash\{j\})) + v(R\backslash\{j\}) - v(R) + b_j.
\end{eqnarray} 
The increase in the utility of the platform is $u(T\cup\{i\}) - u(T) = u_i(T)$, where $u_i(T)$
\begin{eqnarray}
\nonumber  &=& v_i(T) - p_i,\\
\nonumber        &=& \chi(\tau_i(T)) - \chi(\tau_i(R\backslash\{j\})) + v(R) - v(R\backslash\{j\}) - b_j.
\end{eqnarray}
Since user $j\in R\backslash T$, $T\subseteq R\backslash\{j\}$.
Therefore, $\chi(\tau_i(T)) - \chi(\tau_i(R\backslash\{j\})) \geq 0$.
Further the Remove Bad Reference Users function ensures that for any $j\in R$, $v(R) - v(R\backslash\{j\}) - b_j \geq 0$.
Therefore $u_i(T) \geq 0$.

Since at the start of the Selection phase $u(T) = 0$ and the addition of a user $i$ to $T$ through the Add User function or the Replace User function yields $u_i(T) \geq 0$, ONLINE-SMART is profitable. \end{proof}

\begin{lem} 
ONLINE-SMART is truthful.
\end{lem}
\begin{proof} From Theorem \ref{Myerson_Theorem}, ONLINE-SMART is truthful if its selection rule is monotone and the payment it makes is critical. 
Let us assume that user $i$ enters the final winners set $T$ with bid $b_i$. 

Let user $i$ change its bid to $b_i'$ with $b_i' < b_i$.
If user $i$ entered $T$ through the Add User function, then $v_i(R) - b_i > 0$.
Since $v_i(R) - b_i' > 0$ user $i$ enters $T$.  
If user $i$ entered $T$ through the Try To Replace Function then for some $j\in R\backslash T$, $v((R\backslash \{j\}) \cup \{i\}) - v(R) + b_j - b_i > 0$.
Further user $i$ replaces user $j$ in $R$.
Since $v((R\backslash \{j\}) \cup \{i\}) - v(R) + b_j - b_i > 0$, user $i$ still replaces user $j$ in $R$ and enters $T$.
Hence the selection rule of ONLINE-SMART is monotone.

Let user $i$ change its bid to $b_i'$ with $b_i' > p_i$.
If user $i$ entered $T$ through the Add User function then $p_i = v_i(R)$.
Since $b_i' > v_i(R)$, $v_i(R) - b_i' < 0$ and user $i$ no longer enters $T$.  
If user $i$ entered $T$ through the Try To Replace Function, then it replaces some $j\in R\backslash T$ and is paid $p_i = v((R\backslash \{j\}) \cup \{i\}) - v(R) + b_j$.
Since $v((R\backslash \{j\}) \cup \{i\}) - v(R) + b_j - b_i' < 0$, user $i$ no longer replaces user $j$ in $R$ and consequently it does not enter $T$.
Hence the payment made by ONLINE-SMART is critical. \end{proof}

\subsection{Discussion}
ONLINE-SMART is a $k$-secretary algorithm equivalent of the SMART algorithm. It works by initially rejecting a few users, whose profiles are used by the offline SMART algorithm to generate a reference 
set that is used to select/reject future users and to decide their respective payments.  The reference set allows the platform to select users with positive marginal utility. Also, since in the online scenario we assumed that each user does not know when it arrives in relation to other users, the payment strategy with ONLINE-SMART that ensures truthfulness is simpler than the SMART.

In general, the "goodness" of any online algorithm is measured by its competitive ratio, i.e. the ratio of the utility of the online algorithm with the utility of the offline algorithm.
For the $k$-secretary problem, the competitive ratio has been shown to be around $1-1/e$, where the first $1/e$ users are rejected, assuming that users arrive uniformly randomly. 
Finding the competitive ratio of the ONLINE-SMART is, however, very challenging, since one can compare the users that are selected with ONLINE-SMART and SMART, but not their payments since there is no direct relation between them. The latter fact does not allow any tractable analytical solution for finding the 
competitive ratio of the ONLINE-SMART, and to understand its behavior with respect to the SMART, we turn to extensive numerical simulations. Not surprisingly, it turns out it is optimal to approximately reject the first $1/3$ users (similar to $k$-secretary problem), to get the best competitive ratio. 

\section{Simulation}

\begin{figure}
\centering
\includegraphics[scale=0.35]{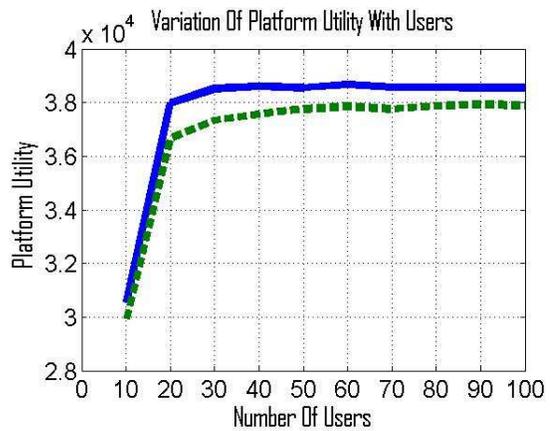}
\caption{Variation of Utility With Users}
\label{fig:UtilityVsUsers}
\end{figure}
In this section, we evaluate the performance of SMART and ONLINE-SMART using numerical simulations and compare it with M-Sensing \cite{M-Sensing2012}. Fig. \ref{fig:UtilityVsUsers} plots platform utility with respect to the number of users.
For Fig. \ref{fig:UtilityVsUsers}, the number of tasks were kept constant while varying the number of users. We used $m=1000$ tasks with values of tasks generated uniformly randomly between $30$ and $50$. 
Each user was assumed to bid for $25$\% of the tasks randomly and each bid was generated uniformly randomly between $5$ and $50$. Fig. \ref{fig:UtilityVsUsers} clearly shows that SMART outperforms M-Sensing by roughly $15$\%.

\begin{figure}
\centering
\includegraphics[scale=0.3]{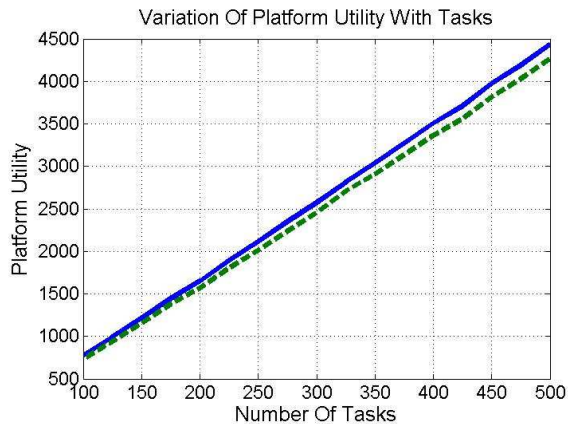}
\caption{Variation of Utility With Tasks}
\label{fig:UtilityVsTasks}
\end{figure}

In Fig. \ref{fig:UtilityVsTasks}, we plot the  platform utility with respect to the number of tasks $m$.
We keep  the number of users fixed to $n=100$. Again, each user was assumed to bid for $25$\% of the tasks randomly and each bid was generated uniformly randomly between $5$ and $50$.
Tasks were generated with values ranging uniformly between $0$ and $20$.

Once again Fig. \ref{fig:UtilityVsTasks} shows that SMART has larger utility compared to M-Sensing, however, the margin is smaller compared to Fig. \ref{fig:UtilityVsUsers}. Further, it is interesting to note that the performance of SMART improves compared to M-Sensing as number of tasks $m$ increases.

\begin{figure}
\centering
\includegraphics[scale=0.25]{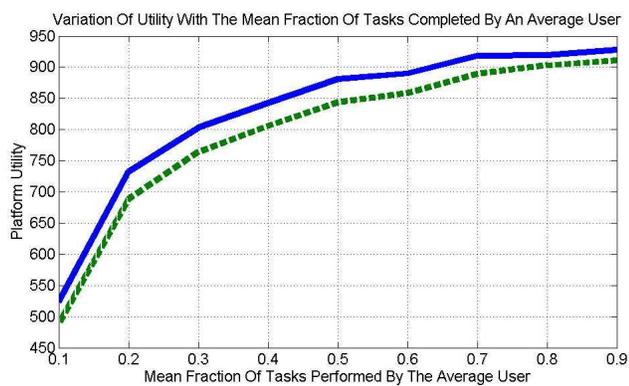}
\caption{Variation of Utility With Mean Task Completion}
\label{fig:UtilityVsMeanTaskCompletion}
\end{figure}

Fig. \ref{fig:UtilityVsMeanTaskCompletion} plots the platform Utility with respect to the mean fraction of tasks that a user completes. The number of users and the number of tasks were kept constant, while the fraction of total tasks that any user completes on average were varied. We used $n=50$ users and $m=100$ tasks, where each user bids uniformly randomly between $5$ and $50$, and each task had values uniformly randomly between $0$ and $20$. From Fig. \ref{fig:UtilityVsMeanTaskCompletion}, we can see that SMART performs better than M-Sensing for all values of fractions of mean task completion.
Further, SMART provides a healthy improvement in Platform Utility over M-Sensing when each user bids to complete about $20$\% to $70$\% of the tasks.  

\begin{figure}
\centering
\includegraphics[scale=0.25]{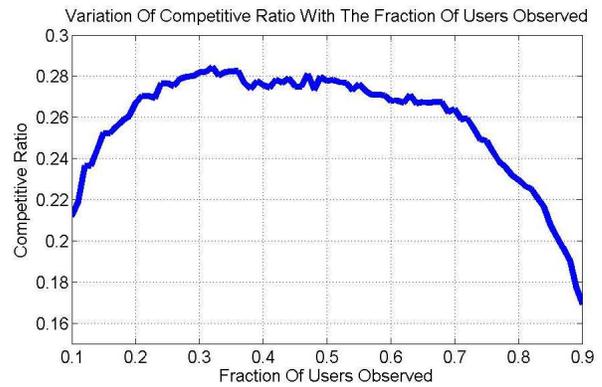}
\caption{Variation of the Competitive Ratio with the Fraction of Users Observed}
\label{fig:CompetitiveRatio}
\end{figure}

Next, we move on to to quantify the performance of the ONLINE-SMART in comparison to SMART to understand the competitive ratio.
In Fig. \ref{fig:CompetitiveRatio}, we plot the competitive ratio of the ONLINE-SMART as a function of $k$. 
We used 
$n=100$ users, and each user was assumed to bid for $25$\% of the tasks randomly and each bid was generated uniformly randomly between $5$ and $50$. The number of tasks $m=30$ with values uniformly random between $0$ and $40$.
Fig. \ref{fig:CompetitiveRatio} indicates that observing $32$\% (around one-third) of the total number of 
users and using their profiles to form a reference set maximizes the competitive ratio of ONLINE-SMART. 

\begin{figure}
\centering
\includegraphics[scale=0.2]{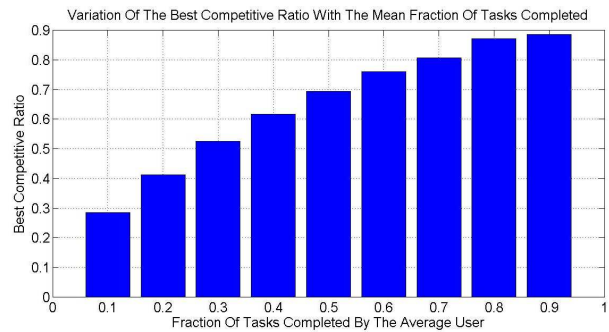}
\caption{Variation of the Competitive Ratio with the Fraction of Total Tasks Completed by an Average User}
\label{fig:CompetitiveRatioMeanFractionTasks}
\end{figure}

Competitive ratio of ONLINE-SMART depends on the fraction of total tasks that each user completes on average.
Larger the fraction, larger is the overlap in the tasks 
completed by different users. Consequently, users arriving in observation phase have a lot of common tasks with the users arriving in selection phase, thereby allowing ONLINE-SMART to make "good" user selections. Thus, it is reasonable to expect that the competitive ratio would increase as the average number of tasks that each user completes increases.
This notion is confirmed through simulation in Fig. \ref{fig:CompetitiveRatioMeanFractionTasks}, where 
the competitive ratio was plotted as a function of the total tasks completed by any user between $10$\% to $90$\%.
At each step, the best competitive ratio was plotted using the best value of $k$ for that particular fraction of total tasks completed by each user.

\begin{figure}
\centering
\includegraphics[scale=0.21]{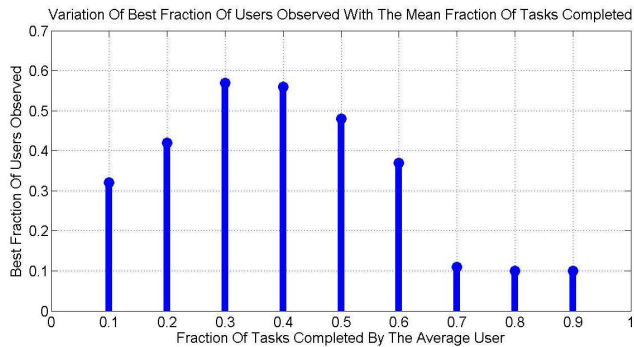}
\caption{Variation of the Fraction Of Users in the Observation Phase with the Fraction of Total Tasks Completed by an Average User}
\label{fig:BestFractionMeanFractionTasks}
\end{figure}

One interesting question remains: given that we know that each user completes a certain fraction of tasks on average, what fraction of users should ONLINE-SMART observe in order to produce the maximum competitive ratio with respect to SMART? This is answered in Fig. \ref{fig:BestFractionMeanFractionTasks}, which is a plot of the fraction of users observed by ONLINE-SMART while producing the competitive ratios depicted in Fig. \ref{fig:CompetitiveRatioMeanFractionTasks}.

We can observe that the fraction of users that need to be observed by ONLINE-SMART has its maximum  when 
each user completes $30$\% of the total tasks on average. This behavior is quite intuitive since, while choosing a certain fraction of users to observe, there is a tradeoff between the quality of the reference set and the number of users left for selection.
If the fraction of users observed is low, the reference set formed in observation phase will be of low quality and hence, the selection of winning users based on the reference set by ONLINE-SMART will be poor. If the fraction of users observed is high, the reference set formed in the observation phase will be of high quality, however, there will be very few users left for selection.

%The conclusion to the paper
\section{Conclusion}
In this paper, we proposed truthful algorithms for mobile crowd-sourcing applications for both the offline and online scenarios. Both the algorithms are inspired by the well-known VCG mechanism that is known to be truthful, with  polynomial complexity compared to exponential/combinatorial complexity of the VCG mechanism. The online version of the algorithm follows the solutions to the $k$-secretary problem, where first some users are just observed, and whose profiles are then used to select users among the remaining ones. Even though we have been able to show that the proposed algorithm have all the four useful properties of any auction design, one question that has not been answered is how much penalty our offline algorithm pay with respect to the optimal offline algorithm, and what is the competitive ratio of our online algorithm. Both of these questions are important in deriving efficient mobile crowd-sourcing auction design.

%The bibliography begins here

%Document ends here

\begin{thebibliography}{10}
\providecommand{\url}[1]{#1}
\csname url@samestyle\endcsname
\providecommand{\newblock}{\relax}
\providecommand{\bibinfo}[2]{#2}
\providecommand{\BIBentrySTDinterwordspacing}{\spaceskip=0pt\relax}
\providecommand{\BIBentryALTinterwordstretchfactor}{4}
\providecommand{\BIBentryALTinterwordspacing}{\spaceskip=\fontdimen2\font plus
\BIBentryALTinterwordstretchfactor\fontdimen3\font minus
  \fontdimen4\font\relax}
\providecommand{\BIBforeignlanguage}[2]{{%
\expandafter\ifx\csname l@#1\endcsname\relax
\typeout{** WARNING: IEEEtran.bst: No hyphenation pattern has been}%
\typeout{** loaded for the language `#1'. Using the pattern for}%
\typeout{** the default language instead.}%
\else
\language=\csname l@#1\endcsname
\fi
#2}}
\providecommand{\BIBdecl}{\relax}
\BIBdecl

\bibitem{Nericell}
``Nericell,'' in \emph{http://research.microsoft.com/en-us/projects/nericell/}.

\bibitem{Sensorly}
``Sensorly,'' in \emph{www.sensorly.com}.

\bibitem{PotholeCrSourcing2008}
J.~Eriksson, L.~Girod, B.~Hull, R.~Newton, S.~Madden, and H.~Balakrishnan,
  ``The pothole patrol: using a mobile sensor network for road surface
  monitoring,'' in \emph{Proceedings of the 6th international conference on
  Mobile systems, applications, and services}.\hskip 1em plus 0.5em minus
  0.4em\relax 2008, pp. 29--39.

\bibitem{EarphoneCrSourcing2010}
R.~K. Rana, C.~T. Chou, S.~S. Kanhere, N.~Bulusu, and W.~Hu, ``Ear-phone: an
  end-to-end participatory urban noise mapping system,'' in \emph{Proceedings
  of the 9th ACM/IEEE International Conference on Information Processing in
  Sensor Networks}.\hskip 1em plus 0.5em minus 0.4em\relax 2010, pp.
  105--116.

\bibitem{CrowdSourcingPractical2009}
T.~Yan, M.~Marzilli, R.~Holmes, D.~Ganesan, and M.~Corner, ``mcrowd: a platform
  for mobile crowdsourcing,'' in \emph{Proceedings of the 7th ACM Conference on
  Embedded Networked Sensor Systems}.\hskip 1em plus 0.5em minus 0.4em\relax
   2009, pp. 347--348.

\bibitem{Sheng2012CrSourcingEnergy}
X.~Sheng, J.~Tang, and W.~Zhang, ``Energy-efficient collaborative sensing with
  mobile phones,'' in \emph{Proceedings of IEEE INFOCOM, 2012}.\hskip 1em plus
  0.5em minus 0.4em\relax  2012, pp. 1916--1924.

\bibitem{WalrandCrSourcing2012}
L.~Duan, T.~Kubo, K.~Sugiyama, J.~Huang, T.~Hasegawa, and J.~Walrand,
  ``Incentive mechanisms for smartphone collaboration in data acquisition and
  distributed computing,'' in \emph{Proceedings of IEEE INFOCOM, 2012}.\hskip 1em
  plus 0.5em minus 0.4em\relax 2012, pp. 1701--1709.

\bibitem{M-Sensing2012}
D.~Yang, G.~Xue, X.~Fang, and J.~Tang, ``Crowdsourcing to smartphones:
  incentive mechanism design for mobile phone sensing,'' in \emph{Proceedings
  of the 18th ACM Annual International Conference on Mobile Computing and
  Networking}.\hskip 1em plus 0.5em minus 0.4em\relax 2012, pp. 173--184.

\bibitem{VCG1961}
\BIBentryALTinterwordspacing
W.~Vickrey, ``Counterspeculation, auctions, and competitive sealed tenders,''
  \emph{The Journal of Finance}, vol.~16, no.~1, pp. 8--37, 1961. [Online].
  Available: \url{http://dx.doi.org/10.1111/j.1540-6261.1961.tb02789.x}
\BIBentrySTDinterwordspacing

\bibitem{VCGIssues}
V.~Conitzer and T.~Sandholm, ``Failures of the vcg mechanism in combinatorial
  auctions and exchanges,'' in \emph{Proceedings of the fifth international
  joint conference on Autonomous agents and multiagent systems}.\hskip 1em plus
  0.5em minus 0.4em\relax  2006, pp. 521--528.

\bibitem{OnlineVCG}
D.~Parkes and S.~Singh, ``An mdp-based approach to online mechanism design,''
  in \emph{Advances in neural information processing systems}, 2003, p. None.

\bibitem{AuctionSurvey}
S.~De~Vries and R.~V. Vohra, ``Combinatorial auctions: A survey,''
  \emph{INFORMS Journal on computing}, vol.~15, no.~3, pp. 284--309, 2003.

\bibitem{NarahariMechanismDesignTutorialI2008}
D.~Garg, Y.~Narahari, and S.~Gujar, ``Foundations of mechanism design: A
  tutorial part 1-key concepts and classical results,'' in \emph{Sadhana
  (Academy Proceedings in Engineering Sciences)}, vol.~33, no.~2.\hskip 1em
  plus 0.5em minus 0.4em\relax Indian Academy of Sciences, 2008, pp. 83--130.

\bibitem{NarahariMechanismDesignTutorialII2008}
------, ``Foundations of mechanism design: A tutorial part 2-advanced concepts
  and results,'' \emph{Sadhana}, vol.~33, no.~2, pp. 131--174, 2008.

\bibitem{Luo2013CrowdSourcing}
K.~{Han}, C.~{Zhang}, and J.~{Luo}, ``{Truthful Scheduling Mechanisms for
  Powering Mobile Crowdsensing},'' \emph{ArXiv e-prints}, Aug. 2013.

\bibitem{LleinbergK-Secretary2005}
R.~Kleinberg, ``A multiple-choice secretary algorithm with applications to
  online auctions,'' in \emph{Proceedings of the sixteenth annual ACM-SIAM
  Symposium on Discrete Algorithms}.\hskip 1em plus 0.5em minus 0.4em\relax
  Society for Industrial and Applied Mathematics, 2005, pp. 630--631.

\bibitem{babaioff2007k-Secretary}
M.~Babaioff, N.~Immorlica, D.~Kempe, and R.~Kleinberg, ``A knapsack secretary
  problem with applications,'' \emph{Approximation, Randomization, and
  Combinatorial Optimization. Algorithms and Techniques}, pp. 16--28, 2007.

\bibitem{BorodinOnlineBook}
A.~Borodin and R.~El-Yaniv, \emph{Online Computation and Competitive
  Analysis}.\hskip 1em plus 0.5em minus 0.4em\relax Cambridge University Press,
  1998.

\bibitem{Myerson1981Auction}
R.~B. Myerson, ``Optimal auction design,'' \emph{Mathematics of operations
  research}, vol.~6, no.~1, pp. 58--73, 1981.

\end{thebibliography}
\end{document}